\documentclass{IEEEtran4PSCC}
\ifCLASSINFOpdf
   \usepackage[pdftex]{graphicx}
\else
   \usepackage[dvips]{graphicx}
\fi
\usepackage[cmex10]{amsmath}

\usepackage{multicol}

\makeatletter
\let\old@ps@headings\ps@headings
\let\old@ps@IEEEtitlepagestyle\ps@IEEEtitlepagestyle
\def\psccfooter#1{%
    \def\ps@headings{%
        \old@ps@headings%
        \def\@oddfoot{\strut\hfill#1\hfill\strut}%
        \def\@evenfoot{\strut\hfill#1\hfill\strut}%
    }%
    \def\ps@IEEEtitlepagestyle{%
        \old@ps@IEEEtitlepagestyle%
        \def\@oddfoot{\strut\hfill#1\hfill\strut}%
        \def\@evenfoot{\strut\hfill#1\hfill\strut}%
    }%
    \ps@headings%
}
\makeatother

\usepackage{todonotes}

\usepackage{enumitem}
\usepackage{verbatim}
\usepackage{multicol}
\usepackage{array}
\usepackage{hyperref}
\usepackage{tikz-3dplot}
\usepackage{multirow}

\usepackage[cmex10]{amsmath}
\usepackage{amsfonts}
\usepackage{amssymb}
\usepackage{amsthm}
\newtheorem{theorem}{Theorem}

\usepackage{accents}
\newcommand{\ubar}[1]{\underaccent{\bar}{#1}}

\newlist{abbrv}{itemize}{1}
\setlist[abbrv,1]{label=,labelwidth=0.5in,align=parleft,itemsep=0.1\baselineskip,leftmargin=!}

\usepackage{booktabs}
\usepackage{cuted} 

\usepackage{float}
\ifCLASSOPTIONcompsoc
 \usepackage[caption=false,font=normalsize,labelfont=sf,textfont=sf]{subfig}
\else
 \usepackage[caption=false,font=footnotesize]{subfig}
\fi

\floatstyle{ruled}
\newfloat{model}{thp}{lop}
\floatname{model}{Model}

\begin{document}
%

\twocolumn

\title{High-Resolution PTDF-Based Planning of Storage and Transmission Under High Renewables}

\author{
\IEEEauthorblockN{Kevin Wu \\ Pascal Van Hentenryck}
\IEEEauthorblockA{H. Milton Stewart School of Industrial and Systems Engineering \\ Georgia Institute of Technology, Atlanta, GA, United States \\
\{kwu381, pvh\}@gatech.edu}
\and
\IEEEauthorblockN{Rabab Haider}
\IEEEauthorblockA{Civil and Environmental Engineering\\
University of Michigan, Ann Arbor, MI, USA\\
rababh@umich.edu}
}


\maketitle

\begin{abstract}
    Transmission Expansion Planning (TEP) optimizes power grid upgrades and investments to ensure reliable, efficient, and cost-effective electricity delivery while addressing grid constraints. To support growing demand and renewable energy integration, energy storage is emerging as a pivotal asset that provides temporal flexibility and alleviates congestion. This paper develops a multiperiod, two-stage PTDF formulation that co-optimizes transmission upgrades and storage siting/sizing. To ensure scalability, a trust-region, multicut Benders scheme warm-started from per-representative-day optima is proposed. Applied to a 2,000-bus synthetic Texas system under high-renewable projections, the method attains final optimality gaps below 2\% and yields a plan with storage at 167 nodes (32\% of peak renewable capacity). These results demonstrate that the proposed PTDF-based methodology efficiently handles large distributed storage fleets, demonstrating scalability at high spatial resolution.
\end{abstract}
\allowdisplaybreaks

\begin{IEEEkeywords}
transmission expansion planning, storage planning, co-optimization, PTDF, renewables.
\end{IEEEkeywords}
\thanksto{\noindent This work was partially funded by NSF award 2112533.}
\section*{Nomenclature}
\label{sec:nomenclature}

\newcommand{\db}{\delta B}
\newcommand{\dbhi}{\bar{\delta B}}
\newcommand{\dblo}{\ubar{\delta B}}
\newcommand{\dva}{{\Delta \theta}}
\newcommand{\dvahi}{\bar{\Delta \theta}}
\newcommand{\dvalo}{\ubar{\Delta \theta}}

\paragraph{Sets}

\renewcommand{\S}{\mathcal{S}}
\newcommand{\N}{\mathcal{N}}
\newcommand{\T}{\mathcal{T}}
\newcommand{\E}{\mathcal{E}}

\begin{abbrv}
    \item[$\N$] Set of buses; $\N = \{1, ..., N\}$; indexed by $i$
    \item[$\E$] Set of branches; $\E = \{1, ..., E\}$; indexed by $ij$
    \item[$\S$] Set of scenarios; $\S = \{1, ..., S\}$; indexed by $s$
    \item[$\T$] Set of hours; $\T = \{1, ..., T\}$; indexed by $t$
\end{abbrv}
\vspace{\baselineskip} 

\paragraph{Parameters}

\newcommand{\pd}{\mathbf{p}^{\text{d}}}
\newcommand{\pgmin}{\mathbf{\ubar{p}}^{\text{g}}}
\newcommand{\pgmax}{\mathbf{\bar{p}}^{\text{g}}}
\newcommand{\pfmax}{\mathbf{\bar{p}}^{\text{f}}}
\newcommand{\dvamin}{\ubar{\theta}}
\newcommand{\dvamax}{\bar{\theta}}

\newcommand{\df}{\delta \pf}
\newcommand{\dX}{\delta X}
\newcommand{\dXmin}{\ubar{\delta X}}
\newcommand{\dXmax}{\bar{\delta X}}
\newcommand{\dB}{\delta B}
\newcommand{\dBmin}{\ubar{\delta B}}
\newcommand{\dBmax}{\bar{\delta B}}

\begin{abbrv}
    \item[$\pd_{i,s,t}$] Load at node $i$ in scenario $s$ at hour $t$
    \item[$\pgmin_{i,s,t}$] min output of generator $i$ in scenario $s$ at hour $t$
    \item[$\pgmax_{i,s,t}$] max output of generator $i$ in scenario $s$ at hour $t$
    \item[$c^{\text{cap}}_{ij}$] Line capacity upgrade cost of edge $ij$
    \item[$c^{\text{stor}}_{i}$] Storage energy rating cost at node $i$
    \item[$c^{\text{G}}_{i}(\cdot)$] Production cost function of generator $i$
    \item[$\Omega_{\S}(\cdot)$] Probability mass function for set of scenarios $\S$
    \item[$\lambda$] Load shedding penalty cost
    \item[$\Delta^c_{ij}$] Capacity upgrade increment of branch $ij$
    \item[$\pfmax_{ij}$] Thermal limit of branch $ij$
\end{abbrv}

\paragraph{Variables}

\newcommand{\pg}{\mathbf{p}^{\text{g}}}
\newcommand{\pf}{\mathbf{p}^{\text{f}}}
\newcommand{\va}{\theta}

\begin{abbrv}
    \item[$\gamma_{ij}$] Capacity upgrade level of branch $ij$
    \item[$\sigma_i$] Storage energy rating at node $i$
    \item[$\text{soc}_{i,s,t}$] Storage SoC at node $i$ in scenario $s$ at hour $t$
    \item[$\text{ch}_{i,s,t}$] Storage charge at node $i$ in scenario $s$ at hour $t$
    \item[$\text{dis}_{i,s,t}$] Storage discharge at node $i$ in scenario $s$ at hour $t$
    \item[$\pg_{i,s,t}$] Output of generator $i$ in scenario $s$ at hour $t$
    \item[$\pf_{ij,s,t}$] Power flow on branch $ij$ in scenario $s$ at hour $t$
    \item[$\xi_{i,s,t}$] Load shed at node $i$ in scenario $s$ at hour $t$
\end{abbrv}
\section{Introduction}
\label{sec:introduction}
The U.S. electric grid is rapidly evolving with rising electricity demand and renewable energy penetration,  requiring substantial transmission infrastructure investments to maintain reliable and affordable service \cite{crossedwires2024}. This presents a critical challenge for system operators who must solve Transmission Expansion Planning (TEP) problems to identify cost-effective strategies for grid upgrades to ensure capability of delivering power from generators to load centers. TEP is modeled as a mixed-integer program (MIP) that is difficult to solve for large-scale systems due to discrete investment decisions. They are made more computationally challenging by including investment decisions for generation \cite{hemmati2013comprehensive}, FACTS devices \cite{Wu_2024_TNEPFACTS}, and energy storage \cite{Sheibani_review}. Given the critical role of storage in providing flexibility under high renewable penetration, this work focuses on TEP with storage.

The TEP+Storage literature typically adopts at least one of three simplifications: (i) reduced spatial granularity; (ii) relaxation of DC power flow to a transport model; or (iii) settings whose economics yield sparse storage deployments. These choices reduce problem dimensionality and computational burden. This work addresses these limits: a PTDF-based, high-resolution model under high-renewable scenarios that endogenously deploys a large distributed storage fleet, creating a more demanding computational setting. A trust-region-stabilized Benders scheme addresses this regime, attaining final optimality gaps below 2\% on a 2,000-bus system with storage at 167 nodes (32\% of peak renewable capacity).

\subsection{Related Work}
High-fidelity TEP+Storage formulations are often limited to small systems due to computational complexity. For example, \cite{Ansari} uses a linearized AC model on a 24-bus system, and \cite{Gan} solves a security-constrained DC model on a 53-bus system; these preserve grid physics but do not readily scale to thousands of buses. 

The use of Benders decomposition (BD) for storage sizing and siting in energy planning has also been extensively studied. For example, \cite{Gan} employs BD within their security-constrained framework; \cite{Xiong} applies BD to storage planning under wind uncertainty on a 24-bus system; \cite{Moradi} develops a Pareto-optimal cut variant for joint generation, transmission, and storage planning on a 24-bus system; and \cite{MacRae} co-plans transmission and storage using BD on a 46-bus Brazilian system. These works demonstrate the effectiveness of decomposition-based approaches for storage planning, but are typically applied on small networks. To reach larger systems, many studies adopt transport relaxations. For example, \cite{zuluaga2024parallel} scales to the 8{,}870-bus CATS using progressive hedging, while \cite{JacobsonPecci} considers the Eastern U.S. aggregated into 19 zones using BD. Storage deployment is also frequently sparse: \cite{Dvorkin} studies merchant storage with a tri-level DC model and, under their economic assumptions and 240-bus case, report profitable deployment at $\leq$7 sites (3\% of peak renewables); while \cite{piansky2025optimizing} solves a DC model with binary storage siting/sizing and line undergrounding on CATS using BD under wildfire risk and report deployments at $\leq$9 sites totaling 2.6 GWh ($\approx$11\% of peak renewables). Collectively, these works demonstrate large-scale feasibility; however, they typically involve reduced flow-model fidelity and/or sparse storage arising from the assumed economics or risk setting, which can alter congestion patterns and investment signals while easing computational burden.

\subsection{Contributions}
To address these limitations, this paper proposes a combined PTDF-based DC formulation with a stabilized (trust-region) Benders approach. The methodology is demonstrated in a high-renewables setting on a 2,000-bus network, thereby addressing simultaneous gaps in spatial resolution, network-fidelity, and storage scalability. The contributions are:
\begin{itemize}
    \item A multiperiod, two-stage PTDF model that co-optimizes line upgrades and storage siting/sizing with hourly dispatch and storage dynamics, preserving DC physics while remaining tractable at scale
    \item A trust-region, multicut Benders scheme, warm-started from per-representative day optima, that mitigates dual degeneracy and oscillatory master-iterates, strengthens cuts, and improves convergence predictability
    \item A multiyear evaluation on a 2,000-bus synthetic Texas system under high-renewable scenarios that endogenously yields a large distributed storage fleet (167 devices, 32\% of peak renewables) while attaining final optimality gaps $\leq$2\%, demonstrating quality of solution.
\end{itemize}

The paper proceeds as follows. §\ref{sec:formulation} formulates the two-stage TEP+Storage model; §\ref{sec:methodology} details the Benders methodology; §\ref{sec:casestudy} describes the Texas case study; §\ref{sec:results} reports optimality gaps, runtime, and investment decisions; and §\ref{sec:conclusion} concludes.
\section{Formulation}
\label{sec:formulation}

PTDF-based DCOPF is well established for TEP \cite{rahmani2016comprehensive}, yet large-scale storage integration within PTDF frameworks remains comparatively unexplored. This section presents the PTDF-based TEP+Storage formulation. The formulation assumes that at most one generator is connected to each bus, without loss of generality.

\subsection{The Two-Stage TEP+Storage PTDF Formulation}
\label{sec:formulation:TNEP}
    The TEP+Storage model is a two-stage program. The \textbf{investment} stage selects transmission upgrades and storage siting/sizing to minimize capital expenditures ($CapEx$). The \textbf{recourse} stage fixes those investments and minimizes operational costs ($OpEx$) over the \textbf{investment period} $\mathcal{H}$, the planning horizon during which $CapEx$ decisions apply, thereby capturing the coupling between long-term planning and short-term operations.

    \subsubsection{Investment decisions and variables}
    \label{sec:formulation:TNEP:variables}

    This study prioritizes capacity upgrades of existing transmission lines over new line additions, as new construction often faces economic, regulatory, and social barriers \cite{GridLab2024Reconductoring}. Upgrades are represented by the discrete decision variable $\gamma_{ij}$ for each branch $ij \in \mathcal{E}$, with allowable levels of enhancement $\Delta^c_{ij}$ (MW per line).

    Storage is modeled as short-duration lithium-ion batteries due to their technological maturity and widespread deployment \cite{MIT2022FutureEnergyStorage}. Sizing decisions are captured by $\sigma_i$, denoting the installed energy capacity (MWh) at node $i$. Each unit has a 4-hour duration, giving a power rating of $\tfrac{1}{4}\sigma_i$. Hourly storage operations are modeled by $\text{ch}_{i,t}$, $\text{dis}_{i,t}$, and $\text{soc}_{i,t}$, representing charge, discharge, and state of charge (SoC), respectively.

    Other key variables include nodal generation dispatch $\pg$ and line power flows $\pf$ (auxiliary via PTDF).
    
    \subsubsection{Objective}
    \label{sec:formulation:TNEP:objective}
    The TEP+Storage model minimizes the sum of $CapEx$ and $OpEx$. The $CapEx$ is a linear function of first-stage decisions $(\gamma, \sigma)$ covering line and storage investments:
    \begin{align}
    \label{eq:TNEP:capex}
        CapEx &= f(\gamma, \sigma) := \textstyle\sum_{ij \in \E} c^{\text{cap}}_{ij} \gamma_{ij} + \textstyle\sum_{i \in \N}{c^\text{stor}_i \sigma_i}
    \end{align}
    Given investments $(\gamma, \sigma)$, the operating cost $OpEx_s(\gamma, \sigma)$ of scenario-$s$ includes daily generation expenses, storage charge and discharge costs, and penalties for load shedding using a per-unit penalty $\lambda$ for load shed $\xi_{i,s,t}$. It is computed over the investment period via scaling factor $\alpha$ (e.g., $\alpha=365$ for one year). The aggregate operating cost $OpEx(\gamma, \sigma)$ is a probability-weighted sum across scenarios $s \in \S$.
    \begin{align}
        & OpEx_s(\gamma, \sigma) := GenEx + StorEx + Penalties \label{eq:TNEP:opex_s} \\
        & = \alpha \textstyle\sum_{i \in \N, t \in \T} 
                    [ c^\text{G}_{i}(\pg_{i,s,t}) + c^\text{S}_{i}(\text{ch}_{i,s,t} + \text{dis}_{i,s,t})
                    + \lambda \;\xi_{i,s,t} ] \nonumber, \\
        & OpEx(\gamma, \sigma) := \textstyle\sum_{s \in \S} \Omega_{\S}(s) \;OpEx_s(\gamma, \sigma) \label{eq:TNEP:opex_all}
    \end{align}

    \subsubsection{Constraints}
        \label{sec:formulation:TNEP:constraints}

        The DC power flow model uses the Power Transfer Distribution Factor (PTDF) matrix $\Phi$, which maps net nodal injections to line flows. For scenario $s$ and hour $t$, let $\text{inj}^{s,t} := [\pg_{i,s,t} + \text{dis}_{i,s,t} - \pd_{i,s,t} - \text{ch}_{i,s,t} + \xi_{i,s,t}]^{N}_{i=1}$ denote the vector of net injections.
    The operational constraints are:
    \begin{align}
        \label{eq:TNEP:ptdf_flow}
        [\pf_{ij,s,t}]^{E}_{ij=1} &= \Phi \times \text{inj}^{s,t} \\
        \label{eq:TNEP:global_power_balance}
        \mathbf{1}^\top \text{inj}^{s,t} &= 0 \\
        \label{eq:TNEP:thermal_limit}
        -\pfmax_{ij} - \gamma_{ij} \Delta^{c}_{ij} & \leq \pf_{ij,s,t} \leq \pfmax_{ij} + \gamma_{ij} \Delta^{c}_{ij} \\
        \label{eq:TNEP:generation:min_max_limits}
        \pgmin_{i,s,t} & \leq \pg_{i,s,t} \leq \pgmax_{i,s,t} \\
        \label{eq:TNEP:state_of_charge}
        \text{soc}_{i,s,t} &= \text{soc}_{i,s,t-1} + \text{ch}_{i,s,t}\eta - \text{dis}_{i,s,t}/\eta \\
        \label{eq:TNEP:initial_charge}
        \text{soc}_{i,s,1} &= \text{ch}_{i,s,1}\eta - \text{dis}_{i,s,1}/\eta \\
        \label{eq:TNEP:final_charge}
        \text{soc}_{i,s,T} &= 0 \\
        \label{eq:TNEP:soc_energy_rating}
        0 & \leq \text{soc}_{i,s,t} \leq \sigma_{i} \\
        \label{eq:TNEP:charge_limits}
        0 & \leq \text{ch}_{i,s,t} \leq \sigma_i/4 \\ 
        \label{eq:TNEP:discharge_limits}
        0 & \leq \text{dis}_{i,s,t} \leq \sigma_i/4
    \end{align}
    where \eqref{eq:TNEP:ptdf_flow} computes the vector of line flows; \eqref{eq:TNEP:global_power_balance} imposes power balance; 
    \eqref{eq:TNEP:thermal_limit} enforces thermal flow limits; 
    \eqref{eq:TNEP:generation:min_max_limits} bounds generator outputs, with renewable generators modeled with $\pgmin=0$ and scenario-dependent $\pgmax$ (curtailment allowed), and conventional units dispatchable within their limits; 
    \eqref{eq:TNEP:state_of_charge}–\eqref{eq:TNEP:discharge_limits} govern storage, including SoC with efficiency losses, a zero-SoC boundary at the start and end of each representative day, and charge/discharge limits. Model \ref{model:TNEP_ptdf} outlines the full two-stage TEP+Storage formulation.

    \begin{model}[!t]
        \caption{Two-Stage TEP+Storage Formulation}
        \label{model:TNEP_ptdf}
        \begin{align*}
            \text{min} \quad 
                & \eqref{eq:TNEP:capex} + \eqref{eq:TNEP:opex_all}
                \\
            \text{s.t.} \quad 
            & \eqref{eq:TNEP:ptdf_flow}-\eqref{eq:TNEP:global_power_balance} , \quad \forall s \in \S, t \in \T \\
            & \eqref{eq:TNEP:thermal_limit},
                \quad \forall ij \in \E, s \in \S, t \in \T \\ 
            & \eqref{eq:TNEP:generation:min_max_limits},
                \quad \forall i \in \N, s \in \S, t \in \T \\
            & \eqref{eq:TNEP:state_of_charge}
                \quad \forall i \in \N, s \in \S, t \in \{2,...,T\} \\
            & \eqref{eq:TNEP:initial_charge}-\eqref{eq:TNEP:final_charge},
                \quad \forall i \in \N, s \in \S \\
            & \eqref{eq:TNEP:soc_energy_rating}-\eqref{eq:TNEP:discharge_limits},
            \quad \forall i \in \N, s \in \S, t \in \T
        \end{align*}
    \end{model}

    \emph{Remark.}
    To maintain recourse feasibility, load shed slack variables $\xi_{i,s,t}$ are included and penalized in $OpEx$ by a large coefficient $\lambda$ (see §\ref{sec:formulation:TNEP:objective}). Any plan with $\xi>0$ is rejected; $\lambda$ is increased (doubled) and the model is re-solved until zero load shed is achieved. This approach preserves feasibility and is guaranteed to terminate when a no-shed solution exists, consistent with industry practice where TEP must serve future load by system design.

\subsection{Scalable PTDF Reformulation}
    \label{sec:formulation:TNEP_PTDF}
    The PTDF model is scaled to high resolution systems using a sparse representation of the PTDF matrix and lazy constraint generation. A PTDF cutoff is a common sparsification approach, where entries of $\Phi$ with magnitudes below a threshold $k_p$ (typically 0.01–0.05) are set to zero, thereby sparsifying the matrix and improving computational efficiency by removing negligible effects \cite{rahmani2016comprehensive}. For very large systems, however, even the sparse PTDF can remain expensive. To address this, the PTDF flow \eqref{eq:TNEP:ptdf_flow} and thermal-limit constraints \eqref{eq:TNEP:thermal_limit} are enforced lazily: the model is first solved without these constraints, then any violated line-flow inequalities are iteratively added to the model until no violations remain. Specifically, if branch $ij$ in scenario $s$ at time $t$ violates its upper bound, the following inequality constraint is added to the model:  
    \begin{align*}
        \pf_{ij,s,t} = \textstyle\sum_{k \in \N} \text{PTDF}_{ij,k} \text{inj}^{s,t}_k \leq \pfmax_{ij} + \gamma_{ij} \Delta^{c}_{ij},
    \end{align*}
    with the inequality reversed for lower bound violations. Lazy constraints preserve model equivalence while omitting nonbinding limits, yielding smaller and more tractable formulations. The set of lazy constraints is further reduced by adding only the $k_v$ most severe violations per iteration, where hyperparameter $k_v$ can be tuned: smaller values produce a sparser model at the cost of more iterations to convergence. Both techniques (cutoff and lazy constraints) are used throughout the remainder of the paper.

\section{Methodology}
\label{sec:methodology}

The two-stage TEP+Storage model grows linearly with the number of second-stage scenarios; combined with the network’s high spatial resolution, even a modest number of scenarios yields very large instances. To address this, a Benders decomposition (BD) is employed. The baseline BD formulation (§\ref{subsec:baseline_benders}) decomposes the problem into a master problem (MP) and subproblems (SPs) linked by investment and operational decisions. Two methodological enhancements are introduced to improve convergence and stability: a warm-start module (§\ref{subsec:relaxed_two_stage}) that initializes the MP using relaxed investment solutions, and a trust-region Benders module (§\ref{subsec:trust_region_benders}) that stabilizes the iterative process using a proximal constraint and restricted candidate set. §\ref{subsec:lower_bound} presents a valid global lower bound used to assess optimality gaps.

\subsection{Multicut Benders for TEP+Storage}
   \label{subsec:baseline_benders}
    The BD algorithm alternates between the MP, which selects first-stage investments to minimize $CapEx$, and scenario SPs, which--given $(\gamma,\sigma)$--optimize continuous operational variables and minimize $OpEx$. Each SP is a linear program, and SPs are solved in parallel. From each SP, feasibility and optimality cuts are generated and added to the MP. Under the multicut scheme, the MP retains all accumulated cuts per scenario. Let $\mathcal{C}_s^{\,k}$ denote the set of cuts for scenario $s$ generated through iteration $k-1$ (with $\mathcal{C}_s^{\,1}=\emptyset$); at iteration $k$ the MP imposes $\mathcal{C}_s^{\,k}$ for all $s\in\mathcal{S}$. Figure \ref{fig:benders} summarizes the algorithm.

    \begin{figure}[!t]
        \centering
        \includegraphics[width=0.75\linewidth]{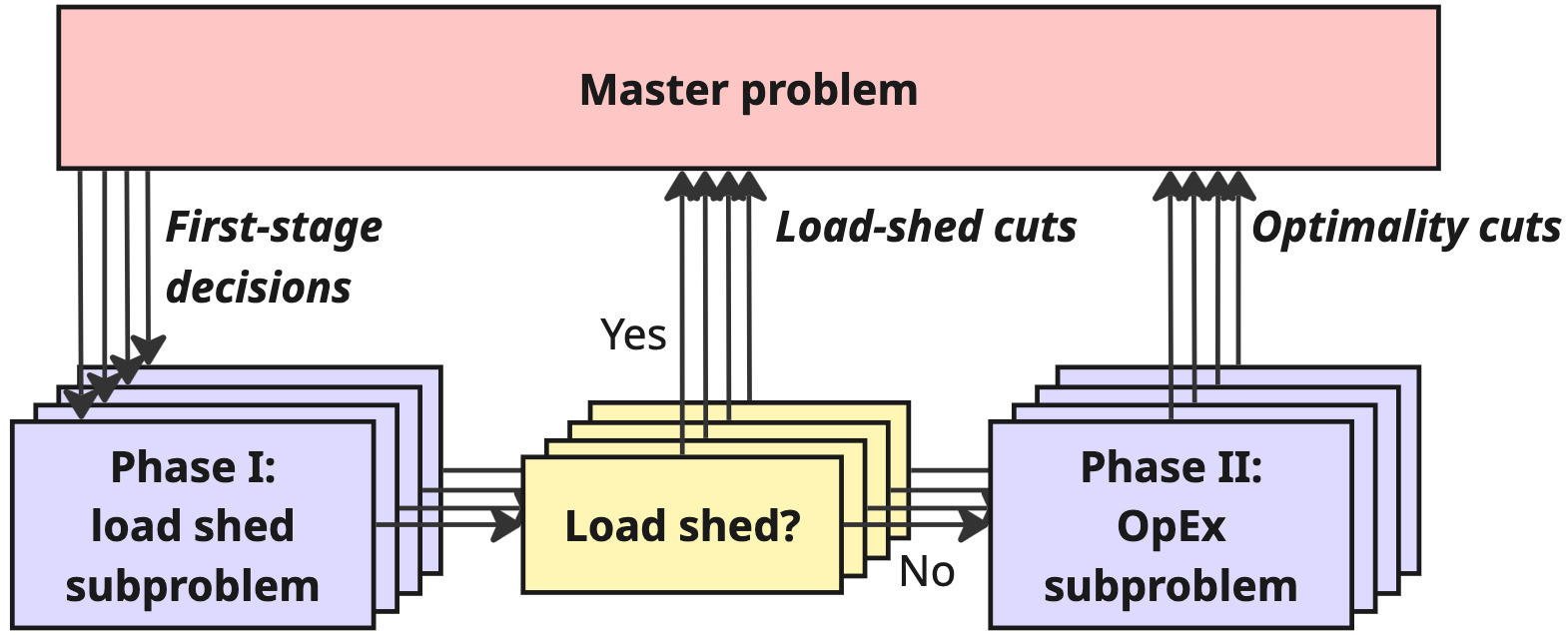}
        \caption{Benders Schematic Flow Diagram.}
        \label{fig:benders}
    \end{figure}

    \subsubsection{Benders MP}
    The MP co-optimizes line and storage investments to minimize $CapEx$. To capture operational and feasibility constraints, the MP includes proxy variables $\theta_s$ an underestimator of $OpEx_s$ for each scenario $s$, and  $\rho$ an underestimator of scenario-wide load-shed. The MP at iteration $k$ is presented in Model \ref{model:TNEP_ptdf_master}, where \eqref{eq:TNEP:no_load_shed} prevents load shedding in the final solution. After solving the MP, the new investment point $(\gamma^{k}, \sigma^{k})$ is broadcast to the SPs, which are solved in parallel to generate the next set of cuts in $\mathcal{C}_s^{\,k+1}$.
    
    \begin{model}[!t]
            \caption{Benders Master Problem (MP)}
            \label{model:TNEP_ptdf_master}
            \begin{align}
                z^{\text{MP}}_k = \text{min} \quad 
                    & \eqref{eq:TNEP:capex} + \textstyle\sum_{s \in \S}\Omega_{\S}(s)\theta_s + \rho \nonumber
                    \\
                \text{s.t.} \quad 
                \label{eq:TNEP:no_load_shed}
                & \rho \leq 0 \\
                & \textstyle\bigcup_{s \in \S}\mathcal{C}_s^k \nonumber
            \end{align}
    \end{model}

   \subsubsection{Benders SP}
   \label{sec:formulation:benders_sub}
   The scenario-$s$ SP optimizes operations under the MP decisions $(\gamma^k,\sigma^k)$, fixed in the SP with dual dual multipliers $\pi^{s,k}_\gamma$ and $\pi^{s,k}_\sigma$ in \eqref{eq:TNEP:fixed_line}-\eqref{eq:TNEP:fixed_stor}. Each SP runs two phases: Phase I enforces load feasibility and Phase II optimizes dispatch. In Phase I, the SP minimizes load shed, $\sum_{i \in \N, t \in \T} \xi_{i,s,t}$. If the optimum exhibits load shed, the SP terminates and an optimality cut is generated for shed underestimator $\rho$, thereafter referred to as a load-shed cut. If no shed occurs, the SP proceeds to Phase II to minimize $OpEx_s$ (Model \ref{model:TNEP_ptdf_subproblem}) and generate an optimality cut for $\theta_s$. Both load-shed and optimality cuts take an affine form, where $\psi \in \{\mathrm{I, II}\}$ indicates the Phase, $\zeta_s = \rho$ in Phase I and $\zeta_s = \theta_s$ in Phase II:
    \begin{align*}
        \zeta_s &\geq \beta^{k}_s + (\pi^{s,k}_\gamma)^\top\gamma + (\pi^{s,k}_\sigma)^\top\sigma \\ 
        \beta^{k}_s &:= v_s^{k,\psi} - (\pi^{s,k}_\gamma)^\top \gamma^k - (\pi^{s,k}_\sigma)^\top \sigma^k \\
        v_s^{k,\mathrm{I}} &= \textstyle\sum_{i,t}\xi_{i,s,t},\qquad
        v_s^{k,\mathrm{II}} = OpEx_s(\gamma^k,\sigma^k)
    \end{align*}
    This two-phase design prioritizes feasibility (no shed) while maintaining complete recourse via $\xi$, avoiding classical infeasibility-ray cuts which can be numerically fragile \cite{bonami2020implementing}.
    
    \begin{model}[!t]
        \caption{Benders Subproblem (SP)}
        \label{model:TNEP_ptdf_subproblem}
        \begin{align}
            \text{min} \quad 
                & OpEx_s(\gamma^k, \sigma^k) \nonumber \\
            \text{s.t.} \quad 
            & \eqref{eq:TNEP:ptdf_flow}-\eqref{eq:TNEP:global_power_balance} , \quad \forall t \in \T \nonumber \\
            & \eqref{eq:TNEP:thermal_limit},
                \quad \forall ij \in \E, t \in \T \nonumber \\ 
            & \eqref{eq:TNEP:generation:min_max_limits},
                \quad \forall i \in \N, t \in \T \nonumber \\
            & \eqref{eq:TNEP:state_of_charge},
                \quad \forall i \in \N, t \in \{2,...,T\} \nonumber \\
            & \eqref{eq:TNEP:initial_charge}-\eqref{eq:TNEP:final_charge},
                \quad \forall i \in \N \nonumber \\
            & \eqref{eq:TNEP:soc_energy_rating} - \eqref{eq:TNEP:discharge_limits},
            \quad \forall i \in \N, t \in \T \nonumber \\
            & \gamma_{ij} = \gamma^k_{ij} \quad (\pi^{s,k}_\gamma), \quad \forall ij \in \E  \label{eq:TNEP:fixed_line} \\
            & \sigma_i = \sigma^k_i \quad (\pi^{s,k}_\sigma), \quad \forall i \in \N  \label{eq:TNEP:fixed_stor}
        \end{align}
    \end{model}

    \subsubsection{Convergence of Benders}
    At iteration $k$, the MP optimum $z_k^{\mathrm{MP}}$ is a valid lower bound $\mathrm{LB}_k$ for Model~\ref{model:TNEP_ptdf}, since the variables $\{\theta_s\}$ underestimate the true scenario costs $\{OpEx_s\}$, giving $\mathrm{LB}_k := z_k^{\mathrm{MP}}$. If all SPs are feasible (Phase I yields $\sum_{i \in \N, t \in \T} \xi_{i,s,t}=0$ for every $s \in \S$), a feasible two-stage objective at $(\gamma^k, \sigma^k)$ is $\hat z_k = f(\gamma^k, \sigma^k) + OpEx(\gamma^k, \sigma^k)$, and is used to maintain a non-increasing upper bound $\mathrm{UB}_k := \min\{\mathrm{UB}_{k-1}, \hat z_k\}$. The algorithm terminates when the relative gap falls below a threshold $\epsilon$: $(\mathrm{UB}_k - \mathrm{LB}_k)/\mathrm{UB}_k \leq \epsilon.$
    
\emph{Remark.} Lazy constraints integrate naturally with BD: per scenario flow/limit constraints identified in one iteration are retained for subsequent iterations, accelerating SP solves.
\subsection{Warm-Start from Relaxed Two-Stage Model (WS Module)}
\label{subsec:relaxed_two_stage}
In the baseline BD formulation, early iterations are feasibility-cut dominated: with no initial cuts, the MP selects a no-expansion plan, causing heavy load shedding in the SPs until a feasible investment set is learned. The WS accelerates this by solving, in parallel for each $s \in \S$, the continuous relaxation of Model \ref{model:TNEP_ptdf} to obtain $(\gamma^{s},\sigma^{s})$. A conservative, non-anticipative design is then formed by component-wise aggregation of the per-scenario optimal investments, and rounded up to the admissible investment granularity:
\begin{align*}
    \bar\gamma := \textstyle\big\lceil \max_{s\in\S}\gamma^{s} \big\rceil, \quad
    \bar\sigma := \textstyle\big\lceil \max_{s\in\S}\sigma^{s} \big\rceil.
\end{align*}
The Benders MP is initialized at $(\bar\gamma,\bar\sigma)$, which reduces early infeasibility cuts and accelerates convergence.
 
\subsection{Trust-Region-Stabilized Benders (TR-Benders Module)}
\label{subsec:trust_region_benders}
Once feasibility is achieved, the BD formulation exhibits strong dual degeneracy--many line-storage configurations with nearly identical operating costs--yielding flat recourse value functions, oscillatory MP iterates, and weak cuts that slow convergence. The TR-Benders addresses this issue by augmenting the BD framework with a proximal trust-region (TR) constraint on first-stage decisions and a restricted candidate set. This stabilization mechanism limits drastic investment changes between iterations, improving numerical stability and ensuring more consistent convergence behavior.

The MP is stabilized around an anchor point $(\breve\gamma,\breve\sigma)$ defined as the incumbent investment that attains the best known upper bound; it is initialized at $(\bar\gamma,\bar\sigma)$ and updated only upon a strict improvement in the upper bound. Inspired by stabilized BD (e.g., \cite{goke2024stabilized}), the method adapts to discrete first-stage variables with an $\ell_1$ trust-region (linear, sparsity-promoting):
\begin{align}
    \textstyle\sum_{ij\in\mathcal{E}} \big|\gamma_{ij}-\breve\gamma_{ij}\big| & \leq r_k^\gamma \\
    \textstyle\sum_{i\in\mathcal{N}} \big|\sigma_{i}-\breve\sigma_{i}\big| & \leq r^\sigma
\end{align}
Absolute values are linearized in the standard way for MIPs. A fixed storage radius $r^\sigma=2$ (permitting limited re-siting) is used, while the transmission radius $r_k^\gamma$ is dynamic: (i) resetting to $1$ whenever the anchor point updates, and (ii) increasing by 1 whenever the MP is infeasible due to the strict level-set constraint in \eqref{eq:strictlevelset}, thereby restoring feasibility of the restricted MP. To further reduce degeneracy, investment decisions are restricted to a candidate set derived from the per-scenario relaxations:
\begin{align*}
    \mathcal{E}_c &:= \{\, ij\in\mathcal{E}\;|\;\exists s\in\S:\ \gamma^{s}_{ij}>0 \} \\
    \mathcal{N}_c &:= \{\, i\in\mathcal{N}\;|\;\exists s\in\S:\ \sigma^{s}_{i}>0 \}
\end{align*}
Finally, a strict level-set constraint prunes regions already known to be suboptimal:
\begin{align}
    \label{eq:strictlevelset}
    f(\gamma,\sigma)+\textstyle\sum_{s\in\S} \Omega_{\S}(s)\,\theta_s < \mathrm{UB}_k
\end{align}

In practice, strict improvement is enforced using a small numerical tolerance. Together, the TR, candidate sets, and level-set constraint damp oscillations and weaken the effects of dual degeneracy. Figure \ref{fig:methodology} summarizes the full workflow of the proposed methodology.

\begin{figure}[!t]
    \centering
    \includegraphics[width=0.7\linewidth]{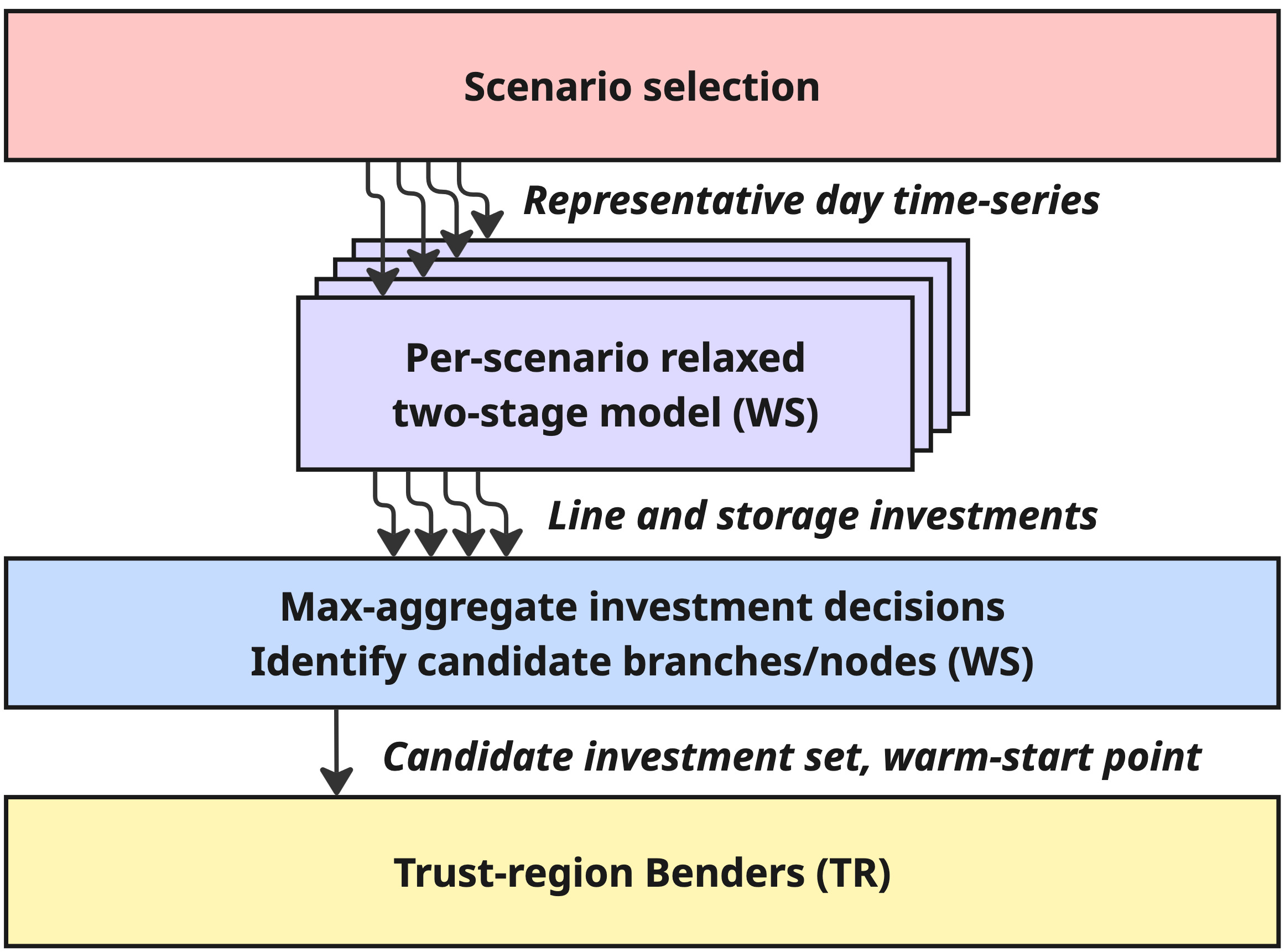}
    \caption{Workflow of the proposed methodology.}
    \label{fig:methodology}
\end{figure}

\subsection{Global Lower Bound and Optimality Assessment}
\label{subsec:lower_bound}
A limitation of the TR stabilization is that the master’s value no longer certifies a global lower bound: with the TR active, $\mathrm{LB}_k$ is only a local lower bound over the current region around the anchor point. However, a valid global lower bound can be constructed as below. 

\subsubsection{First-Stage Lower Bound}
With load shed disallowed by \eqref{eq:TNEP:no_load_shed}, the minimum feasible investment decision is constrained by the first-stage investments of the most challenging representative day (typically characterized by high load and/or low renewable generation). This is formalized in Theorem \ref{thm:first_stage_bound}. 

\begin{theorem}
\label{thm:first_stage_bound}
For each $s\in\S$, define the day-$s$ feasible set and the all-scenarios feasible set:
\begin{align*}
    \Sigma^{s}_{\mathrm{feas}} &:= \{(\gamma,\sigma): \text{recourse is feasible for scenario } s\} \\
    \Sigma_{\mathrm{feas}} &:= \textstyle\bigcap_{s\in\S}\Sigma^{s}_{\mathrm{feas}}
\end{align*}

Given the feasible sets, define the corresponding minimum first stage costs for feasible day-s and all-scenarios:
\begin{align*}
    c^{\mathrm{feas}}_s &:= \min_{(\gamma,\sigma)\in \Sigma^{s}_{\mathrm{feas}}} f(\gamma,\sigma), \quad
    c^{\mathrm{feas}} := \min_{(\gamma,\sigma)\in \Sigma_{\mathrm{feas}}} f(\gamma,\sigma),
\end{align*}
and first-stage optimal investments for scenario $s$:
\begin{align*}
    (\gamma^s_{\mathrm{feas}},\sigma^s_{\mathrm{feas}}) &\in \textstyle\arg\min_{(\gamma,\sigma)\in \Sigma^{s}_{\mathrm{feas}}} f(\gamma,\sigma)
\end{align*}
Let $\hat s \in \arg\max_{s\in\S} c^{\mathrm{feas}}_s$. Then
\begin{align*}
    \textstyle\max_{s\in\S} c^{\mathrm{feas}}_s \leq f(\gamma,\sigma) \quad \forall\; (\gamma,\sigma) \in \Sigma_{\mathrm{feas}}
\end{align*}
In particular, $f(\gamma^{\hat s}_{\mathrm{feas}},\sigma^{\hat s}_{\mathrm{feas}})=\max_{s} c^{\mathrm{feas}}_s$ is a valid lower bound on any year-feasible first-stage plan's cost.
\end{theorem}
\begin{proof}
Since $\Sigma_{\mathrm{feas}}\subseteq \Sigma^{s}_{\mathrm{feas}}$ for every $s$, 
\begin{align*}
c^{\mathrm{feas}}=\min_{(\gamma,\sigma)\in \Sigma_{\mathrm{feas}}} f(\gamma,\sigma)\ \geq \min_{(\gamma,\sigma)\in \Sigma^{s}_{\mathrm{feas}}} f(\gamma,\sigma)=c^{\mathrm{feas}}_s
\quad\forall s,
\end{align*}
yielding $\max_s c^{\mathrm{feas}}_s \leq c^{\mathrm{feas}}$. Choosing $\hat s\in\arg\max_s c_s$ gives the stated bound with $(\gamma^{\hat s}_{\mathrm{feas}},\sigma^{\hat s}_{\mathrm{feas}})$.
\end{proof}

\subsubsection{Second-Stage Lower Bound}
The per-scenario operational cost is monotone nonincreasing in investment capacity, as per Theorem \ref{thm:opex_monotonicity}. Let $(\gamma^{\max},\sigma^{\max})$ denote the componentwise upper bounds on line and storage capacities (not restricted to the candidate set) of Model \ref{model:TNEP_ptdf}. Then a valid lower bound on the aggregated second-stage cost is $\sum_{s\in\S}\Omega_{\S}(s)\, q_s(\gamma^{\max},\sigma^{\max})$.

\begin{theorem}
    \label{thm:opex_monotonicity}
    Let $q_s(\gamma, \sigma)$ be the optimal value of Model \ref{model:TNEP_ptdf_subproblem}. If $\gamma' \succeq \gamma$ and $\sigma' \succeq \sigma$ (componentwise), then $q_s(\gamma',\sigma') \leq q_s(\gamma,\sigma).$
    In particular, $q_s(\cdot)$ is monotone nonincreasing in $\gamma$ and $\sigma$.
\end{theorem}
\begin{proof}
    For fixed $s \in \S$, increasing $\gamma$ relaxes the thermal limit constraints \eqref{eq:TNEP:thermal_limit}, and increasing $\sigma$ relaxes the storage constraints \eqref{eq:TNEP:soc_energy_rating}-\eqref{eq:TNEP:discharge_limits}. Hence the feasible region under $(\gamma', \sigma')$ contains that under $(\gamma, \sigma)$. Because the SP objective is minimized and depends on $\gamma,\sigma$ only through these constraints, enlarging the feasible set cannot increase the optimal value, yielding the claim.
\end{proof}

\subsubsection{Combined Bound}
Define $z^*$ as the optimal final solution's objective, with optimal line and storage investments $(\gamma^*, \sigma^*) \in \Sigma_{\mathrm{feas}}$.
\begin{align*}
    z^* := f(\gamma^*, \sigma^*) + \textstyle\sum_{s\in\S}\Omega_{\S}(s)\, q_s(\gamma^*,\sigma^*)
\end{align*}
Since Theorem \ref{thm:first_stage_bound} and \ref{thm:opex_monotonicity} show, respectively, that
\begin{align*}
    f(\gamma^{\hat s}_{\mathrm{feas}},\sigma^{\hat s}_{\mathrm{feas}}) & \leq  f(\gamma^*, \sigma^*),\\
    q_s(\gamma^{\max},\sigma^{\max}) & \leq q_s(\gamma^*,\sigma^*) \quad \forall \; s \in \S,
\end{align*}
a valid combined lower bound can be constructed:
\begin{align*}
    \text{LB}^* := f(\gamma^{\hat s}_{\mathrm{feas}},\sigma^{\hat s}_{\mathrm{feas}}) + \textstyle\sum_{s\in\S}\Omega_{\S}(s)\, q_s(\gamma^{\max},\sigma^{\max}) \leq z^*
\end{align*}

\subsubsection{Optimality Gap}
Each $c^{\text{feas}}_s$ is computed under the integer relaxation of the first-stage investments (as in the WS). Being a relaxation, this provides a valid lower bound for the year-wide plan, which must satisfy discretized first-stage decisions. The reported optimality gaps are therefore conservative upper bounds: $(\mathrm{UB}_k - \mathrm{LB}^*)/\mathrm{UB}_k$.

\section{Case study}
\label{sec:casestudy}
The proposed methodology is tested on a large-scale power system case study: the ACTIVSg2000 test case, a 2000-bus model of the U.S. Texas grid. The TEP+Storage model is applied in a rolling-horizon manner. The formulation itself has year-agnostic investment variables, but is solved sequentially at 5-year decision points $\mathcal{Y} = \{2030, 2035, 2040, 2045\}$ with investment periods of length $\alpha=1825$ days. Investments (and their capital costs) are cumulative: capacity installed in an earlier year remains available in subsequent years, ensuring monotone, non-decreasing expansion across planning periods. This application reflects practical staging, but it is not guaranteed optimal relative to an omniscient, all-years model; quantifying any cross-period optimality gap is left for future work.

\paragraph{Network Data}
\label{sec:casestudy:data}
The synthetic Texas system is built from public information and statistical analysis of real power systems \cite{synthetic_1}. The system has 2000 buses, 3206 lines, 423 non-renewable generators, and 175 renewable generators. Figure \ref{fig:generators} shows the location of solar, wind, and nonrenewable generators, overlaid on the grid topology. 
\begin{figure}[!t]
\centering
\includegraphics[width=0.75\columnwidth, trim={0 9 0 9}, clip]{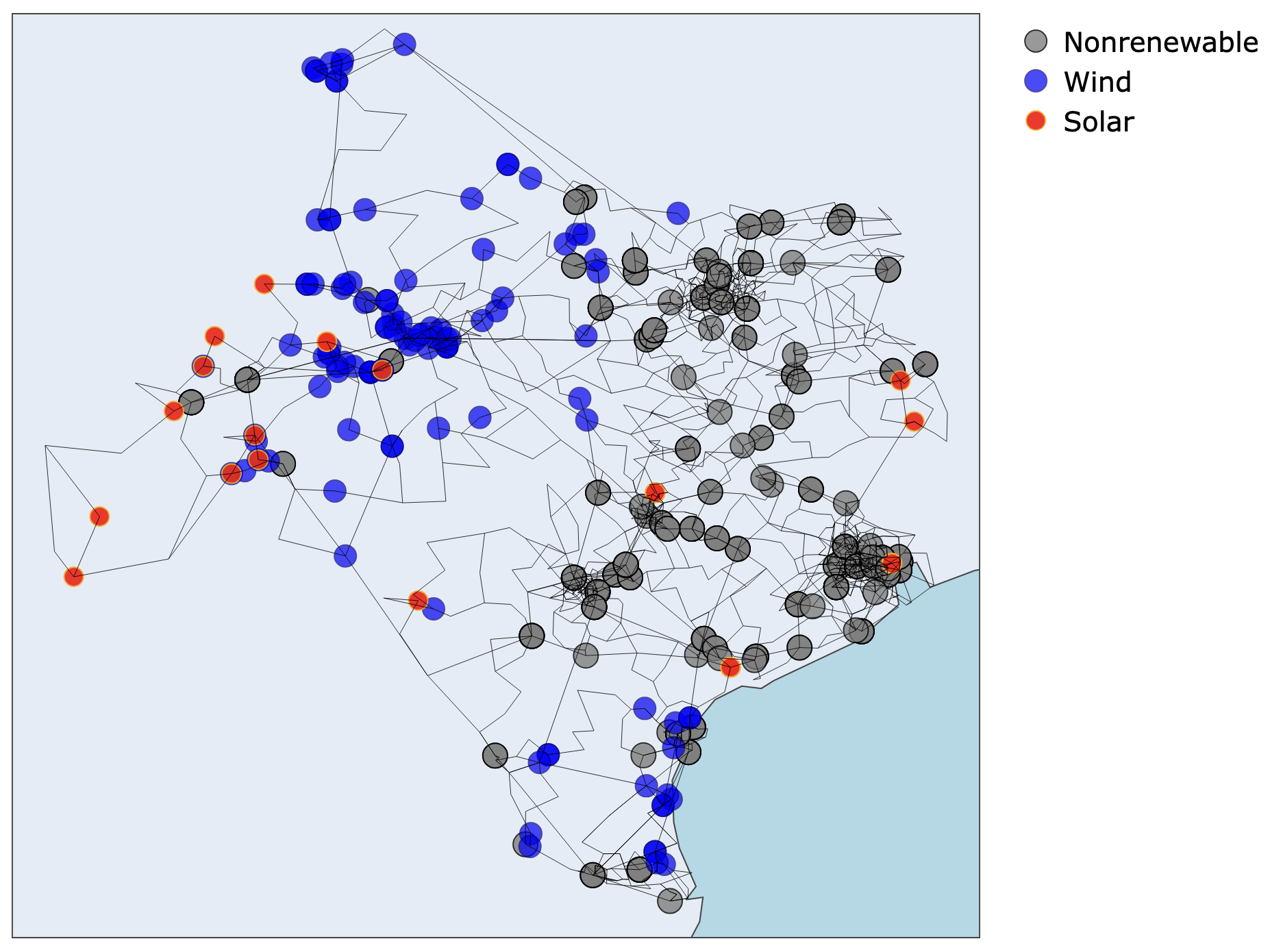}
\caption{Generator sites of the system. \textcolor{black}{\textbf{Black}} circles indicate nonrenewable, \textcolor{blue}{\textbf{blue}} circles indicate wind, and \textcolor{red}{\textbf{red}} circles indicate solar generators.}
\label{fig:generators}
\end{figure}
The future generation mix from 2030 to 2045 follows EIA outlook projections \cite{EIA2023Outlook}. All generation resources are uniformly scaled to reflect projected growth or planned decommissioning. Load growth is modeled with a year-to-year increase of 1.5\% and is scaled uniformly across all buses. Table \ref{tab:gen_load_scale} shows the total generator capacities (GW) of the network and summarizes the scaling for each generation type and load relative to the base year, 2022. 

\begin{table}[!t]
        \centering
        \caption{Generation Base Capacity (GW) and Generation/Load Scaling}
        \label{tab:gen_load_scale}
        \begin{tabular}{lcccccc}
            \toprule
            \textbf{Type} & \textbf{Capacity} & \textbf{2022} & \textbf{2030} & \textbf{2035} & \textbf{2040} & \textbf{2045} \\
            \midrule
            Coal & 14.5 & 1.00 & 0.82 & 0.82 & 0.82 & 0.82 \\
            Natural Gas & 75.1 & 1.00 & 0.79 & 0.73 & 0.71 & 0.72 \\
            Nuclear & 7.1 & 1.00 & 0.98 & 0.90 & 0.80 & 0.80 \\
            Solar & 2.5 & 1.00 & 4.51 & 6.00 & 6.87 & 8.04 \\
            Wind & 15.5 & 1.00 & 2.02 & 2.23 & 2.26 & 2.32 \\
            Load & - & 1.00 & 1.13 & 1.21 & 1.31 & 1.41 \\
            \bottomrule
        \end{tabular}
    \end{table}

\paragraph{Scenario selection}
Scenarios are 24-h representative days sampled from a year of load/wind/solar data. This reduces the computational time required for evaluating a full year of operations, while still capturing intertemporal storage dynamics. In this case study, $k_r = 18$ days are selected by a max-min diversity criterion on normalized daily profiles--maximizing the minimum pairwise Euclidean distance--and assigned equal weights. The selected days capture both high-demand and high-renewable operating regimes that drive investment decisions. Days with higher demand levels typically correspond to days where storage operation is needed to avoid load shedding. Among these, August 11 (08-11) represents the most challenging operating condition as it exhibits the highest average demand. This representative day drives significant investment requirements, as discussed in §\ref{sec:results}.

\paragraph{Investments parameters}

Line upgrade costs are set to \$1243/MW-km, adapted from \cite{Esmaili}. Storage investment cost is set to \$1{,}000{,}000/MWh, with a charge/discharge cost of \$68.50/MWh, adapted from \cite{NREL_ATB_UtilityScaleBattery_2024}. Line upgrades are discretized to capacity increases of 30, 60, and 90\%, reflecting the enhancements achievable through typical line reconductoring \cite{GridLab2024Reconductoring}. Storage units are 250 MWh each, limited to 12 per node (3 GWh total), with charge and discharge efficiencies of $\eta = 95\%$.

\paragraph{Computational Setup and Hyperparameters}
\label{sec:casestudy:computational_specs}
The TEP+Storage model is implemented in Julia with JuMP \cite{Lubin2023} and solved with Gurobi \cite{gurobi} on Georgia Tech’s PACE Phoenix cluster \cite{PACE} (24-core Intel Xeon, Linux); each run used 24 threads and 288 GB RAM. TR-Benders employed a 0.1\% MP optimality-gap tolerance, with wall-clock limits of 24 hours for 2030/2035/2040 and 72 hours for 2045 to accommodate longer per-iteration solves. The per-iteration added-violations quota was set to $k_v=32$ and the PTDF cutoff to $k_p = 0.005$, reducing nonzeros in $\Phi$ by approximately 75\%. Lower PTDF cutoff values (e.g., $k_p=0.001$ or no cutoff) were also tested and substantially increased PTDF density and subproblem solve times, particularly in later planning periods, while having negligible impact on the resulting investment decisions.
\section{Results}
\label{sec:results}

With the proposed methodology, the year-wise TEP+Storage planning problem is both tractable and achieves optimality guarantees of $\leq 2 \%$. Results are organized as follows: optimality gaps (§\ref{sec:results:opt_gaps}), computational performance (§\ref{sec:results:computational}), and investment outcomes (§\ref{sec:results:investments}).

\subsection{Optimality Gaps}
\label{sec:results:opt_gaps}
All optimality gaps are computed relative to the lower bound defined in §\ref{subsec:lower_bound}, which bounds the best achievable objective for each 5-year decision given previously committed investments. Final gap values are reported in the rightmost column of Table \ref{tab:investments_opt_gap}. Across all investment periods, the WS initializes near optimality, yielding initial gaps within 11\%. The subsequent TR-Benders reduces these gaps to below 2\% in every period, validating the approach and the investment solutions presented below.

\begin{table*}[!t]
    \caption{Investment decisions and costs for 2030-2045 for most challenging rep. day, WS, and TR-Benders plans.}
    \centering
    \label{tab:investments_opt_gap}
    \begin{tabular}{clccccccccccc}
        \toprule
         & \multirow{2}{*}{Mode} & \multirow{2}{*}{\# Lines} & \multirow{2}{*}{\begin{tabular}[c]{@{}c@{}}\# Storage units / \\ Total capacity (GWh)\end{tabular}} & \multicolumn{2}{c}{$CapEx$} & \multicolumn{2}{c}{$OpEx$} & \multirow{2}{*}{Opt. Gap} \\
         & & & & Lines (\$B) & Storage (\$B) & Gen (\$B) & StorOpEx (\$B) & & \\
        \midrule 
         \parbox[t]{2mm}{\multirow{3}{*}{\rotatebox[origin=c]{90}{2030}}} 
        & 08-11 & 183 & 0 / 0 & 0.26 & 0 & 46.49 & 0 & - \\
        & Warm-start & 291 & 0 / 0 & 1.07 & 0 & 28.80 & 0 & 3.6\% \\
        & \textbf{TR-Benders} & 191 & 0 / 0 & 0.35 & 0 & 29.05 & 0 & \textbf{2.0\%} \\
        \midrule 
         \parbox[t]{2mm}{\multirow{3}{*}{\rotatebox[origin=c]{90}{2035}}} 
        & 08-11 & 315 & 34 / 28.91 & 0.30 & 28.93 & 50.92 & 7.27 & - \\
        & Warm-start & 375 & 34 / 33.75 & 1.55 & 33.75 & 30.40 & 0.40 & 8.9\% \\
        & \textbf{TR-Benders} & 349 & 30 / 29.50 & 1.06 & 29.50 & 30.47 & 0.41 & \textbf{2.0\%} \\
        \midrule 
         \parbox[t]{2mm}{\multirow{3}{*}{\rotatebox[origin=c]{90}{2040}}} 
        & 08-11 & 419 & 116 / 98.10 & 1.29 & 98.16 & 56.59 & 24.71 & -\\
        & Warm-start & 457 & 116 / 112.75 & 2.48 & 112.75 & 33.35 & 1.64 & 10.6\% \\
        & \textbf{TR-Benders} & 443 & 98 / 98.75 & 1.40 & 98.75 & 33.45 & 1.65 & \textbf{0.7\%} \\
        \midrule 
         \parbox[t]{2mm}{\multirow{3}{*}{\rotatebox[origin=c]{90}{2045}}} 
        & 08-11 & 479 & 175 / 167.31 & 1.58 & 167.31 & 62.66 & 42.02 & -\\
        & Warmstart & 509 & 175 / 183.50 & 2.64 & 183.50 & 36.63 & 3.66 & 7.7\% \\
        & \textbf{TR-Benders} & 491 & 167 / 167.75 & 1.59 & 167.75 & 36.72 & 3.70 & \textbf{0.4\%} \\
        \bottomrule
    \end{tabular}
\end{table*}

Figure \ref{fig:optimality_over_time} plots the evolution of the optimality gap for the 2040 planning period. Only iterations that improve the upper bound, corresponding to zero load-shedding solutions, are displayed. Similar monotonic reductions in the gap are observed across all planning periods, illustrating the effectiveness of the trust-region stabilization in promoting consistent convergence behavior.

\begin{figure}[!h]
\centering
\includegraphics[width=0.75\columnwidth, trim={0 0 0 7}, clip]{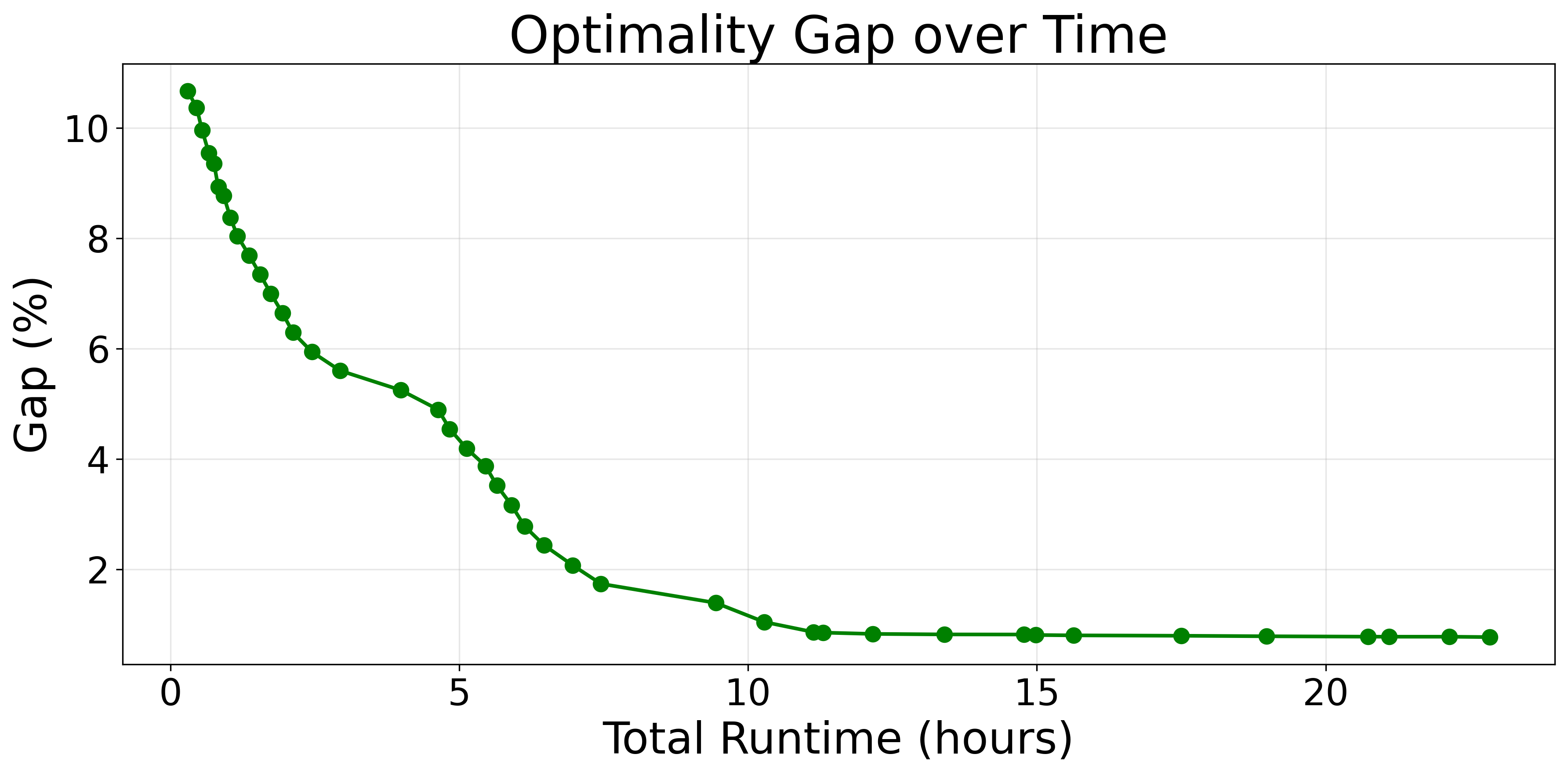}
\caption{Optimality gap over time of TR-Benders for period 2040}
\label{fig:optimality_over_time}
\end{figure}

\subsection{Computational Results}
\label{sec:results:computational}
Under the scalable PTDF reformulation (§\ref{sec:formulation:TNEP_PTDF}), the TEP+Storage model remains computationally tractable at 2,000-bus scale. A tunable PTDF cutoff $k_p$ provides planners with a direct accuracy-scalability trade-off (less straightforward in phase-angle DC formulations). Across all representative days, less than 2.2\% of PTDF flow constraints are activated under the lazy constraint procedure during the WS, and less than 3.2\% during the TR-Benders, representing a substantial reduction in subproblem size.

\subsubsection{WS Module}
 Figure \ref{fig:two_stage_boxplot} summarizes computation times for the relaxed two-stage model in the warm-start module, reported per representative day and 5-year investment period. Runtime is dominated by the iterative discovery of violated thermal-limit constraints; consequently, more congested days take longer to solve. Average time-to-solve falls from 2030 to 2040 because investments made in earlier 5-year periods remain in place in subsequent periods, relieving congestion and simplifying decisions for most representative days. Outliers persist on days that combine heavy storage utilization with elevated congestion (e.g., 08-11 with its high average demand).

\begin{figure}[!t]
\centering
\includegraphics[width=0.75\columnwidth, trim={0 0 0 7}, clip]{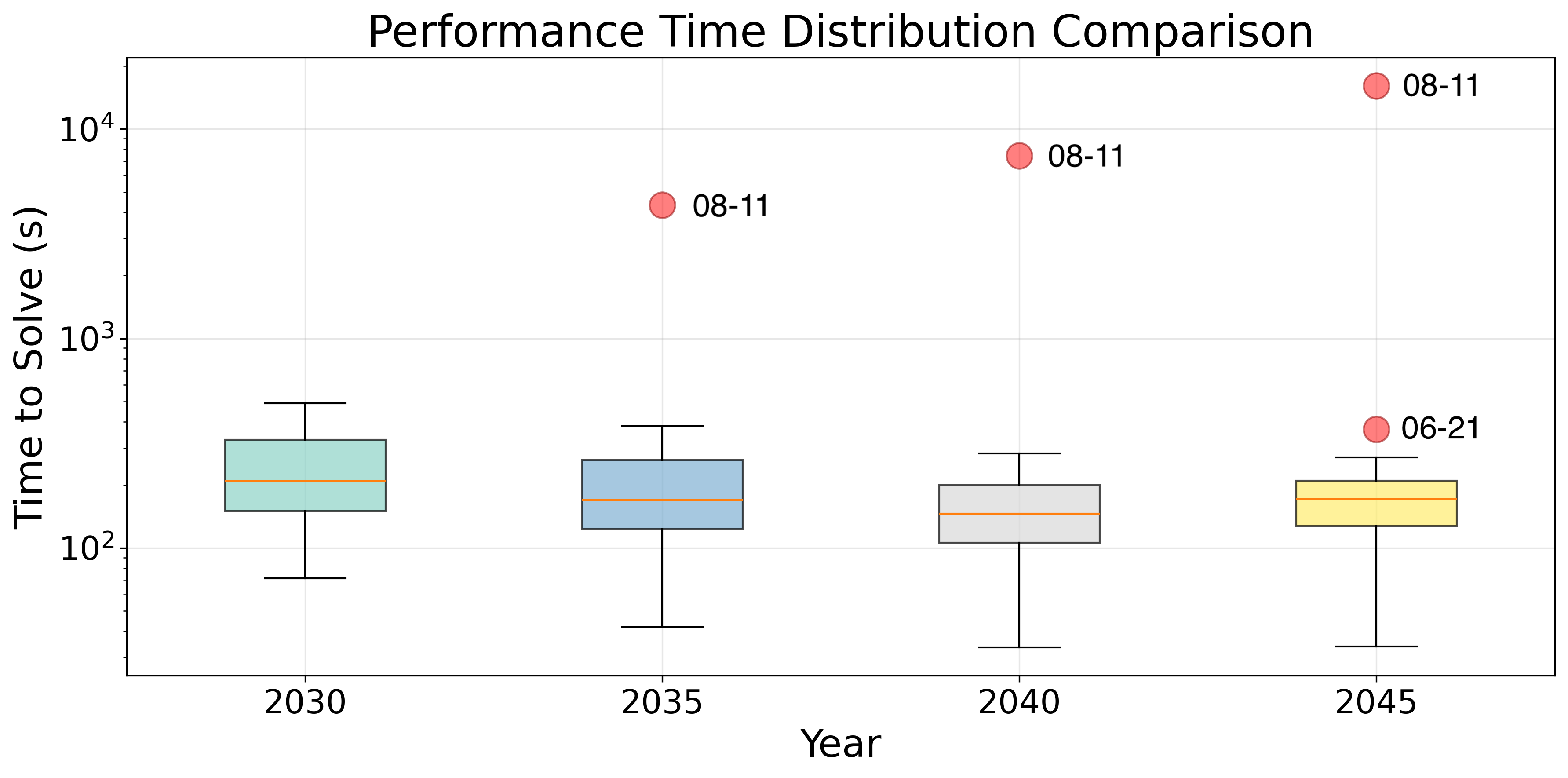}
\caption{WS computation times by representative day for each 5-year investment period.}
\label{fig:two_stage_boxplot}
\end{figure}

\subsubsection{TR-Benders Module}
The TR-Benders demonstrates improved tractability by decoupling first-stage investment decisions from per-scenario operations. Table \ref{tab:benders_time_iter} reports time limits, iteration counts, and average seconds per iteration for each 5-year investment period. Iteration time increases over the planning horizon because progress each round is paced by the most challenging SP/representative day: the algorithm advances only after all cuts are generated. 

To further illustrate this effect, the final TR-Benders investment configuration was fixed and the operational SP was solved for each representative day, as shown in Figure \ref{fig:storage_runtime}. The results show that computational effort increases with storage energy throughput, defined as the total GWh charged and discharged during the representative day, across representative days and planning periods. Higher storage activity increases the number of feasible charging patterns interacting with network constraints, which prolongs SP solves due to storage time-coupling and the ongoing discovery of violated thermal-limit constraints. Further computational improvements may be attained by pre-screening thermal violations.

Overall, the computational performance and low optimality gaps underscores the tractability of the proposed formulation at high spatial and storage resolution.

\begin{figure}[!t]
\centering
\includegraphics[width=0.75\columnwidth, trim={0 0 0 7}, clip]{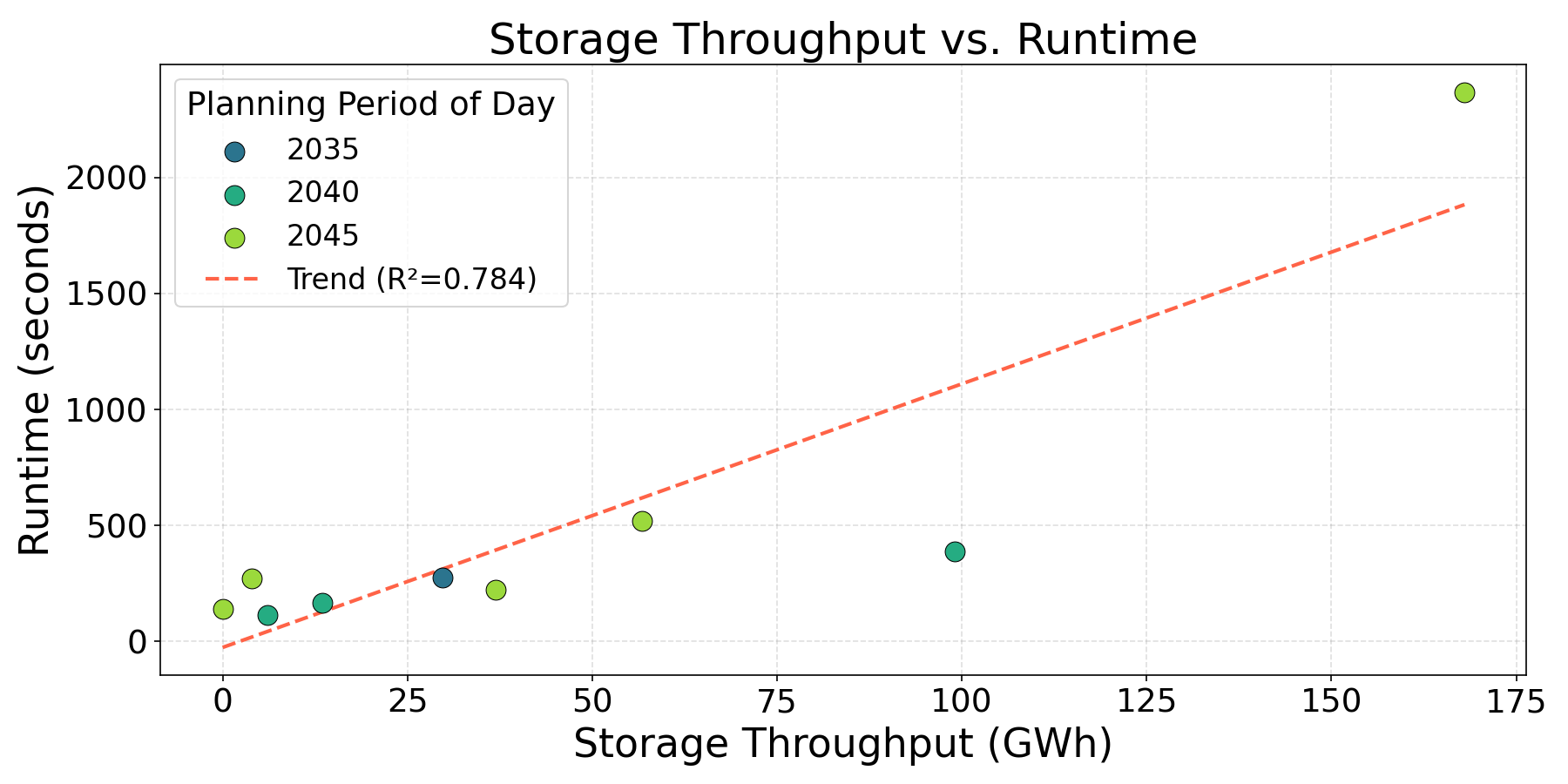}
\caption{SP runtime vs. storage energy throughput (GWh) for representative days with nonzero utilization across the 2030-2045 planning horizon. Values reflect the representative day under the investment plan in effect for that year.}
\label{fig:storage_runtime}
\end{figure}

\begin{table}[!t]
    \centering
    \caption{TR-Benders computational performance for each 5-year investment period.}
    \label{tab:benders_time_iter}
    \begin{tabular}{lcccc}
        \toprule
         & \textbf{2030} & \textbf{2035} & \textbf{2040} & \textbf{2045} \\
        \midrule
        Time Limit (Hrs.) & 24 & 24 & 24 & 72 \\
        Iterations & 1186 & 518 & 190 & 49 \\
        Avg. Seconds/Iter & 73 & 167 & 455 & 5290 \\
        \bottomrule
    \end{tabular}
\end{table}

\begin{figure*}[!t]
\centering
\subfloat[2030 TEP+Storage]{\includegraphics[width=0.45\columnwidth, trim={375 85 380 65}, clip]{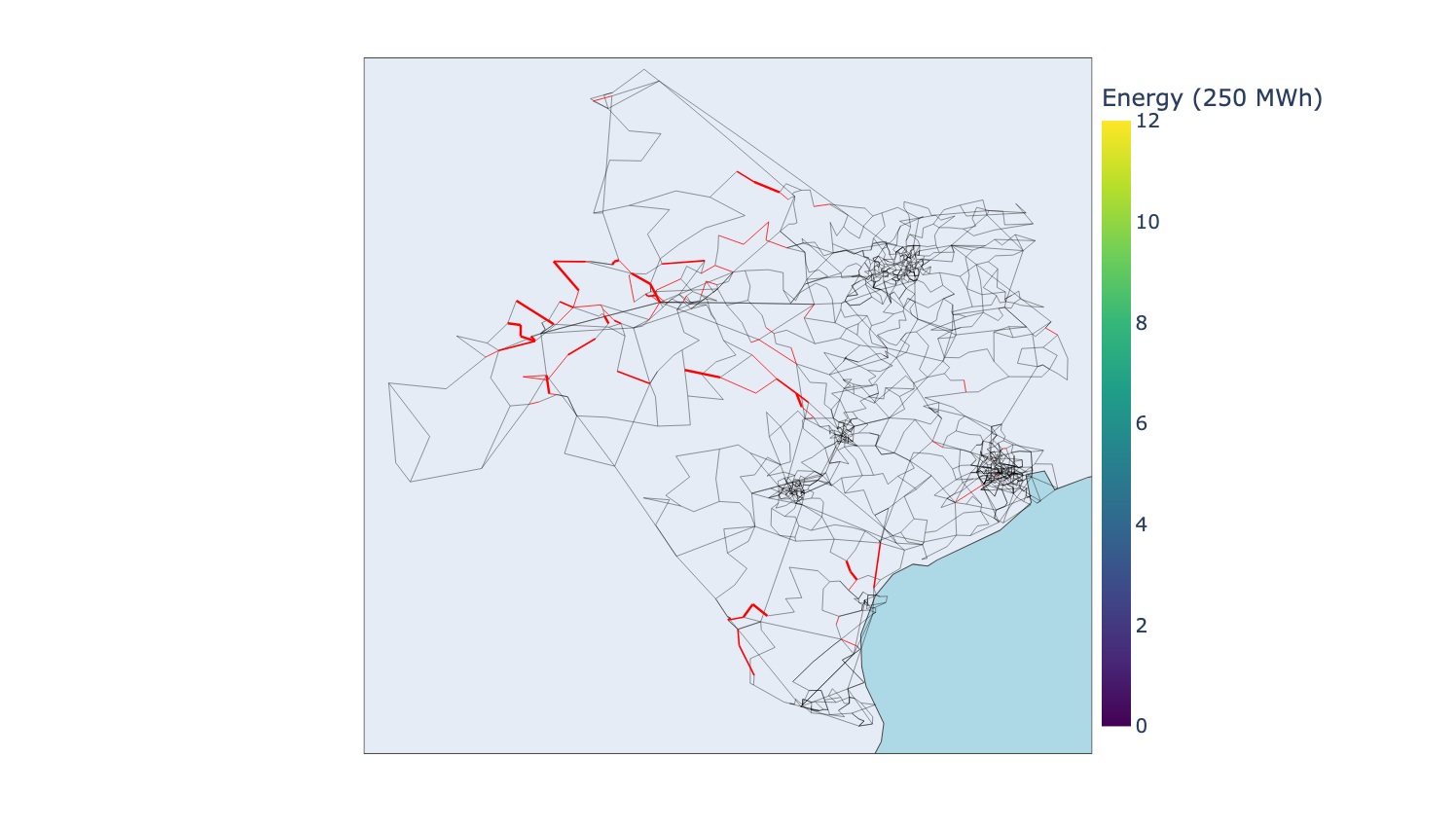}}
\hspace{0.1cm}
\subfloat[2035 TEP+Storage]{\includegraphics[width=0.45\columnwidth, trim={375 85 380 65}, clip]{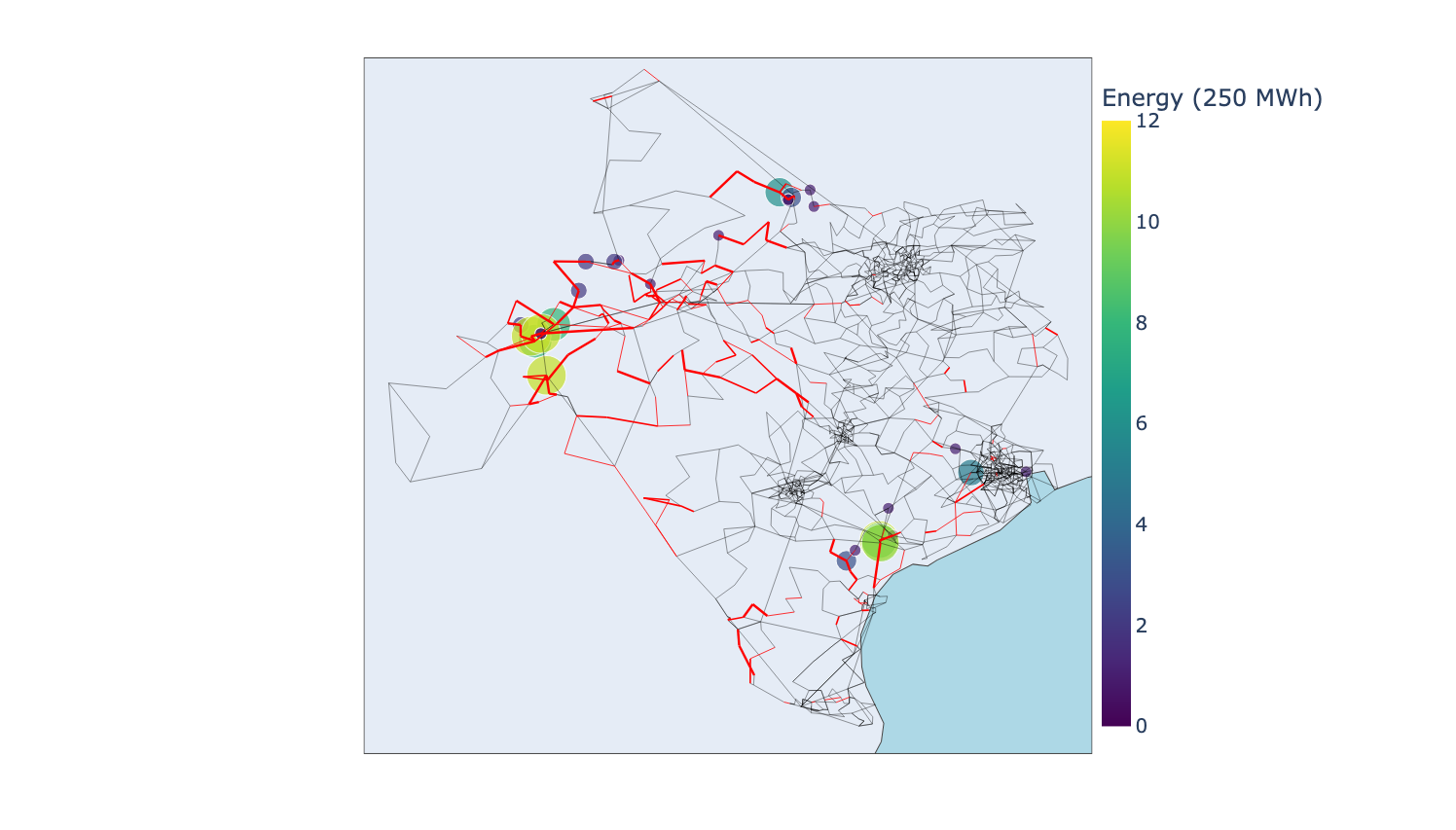}}
\hspace{0.1cm}
\subfloat[2040 TEP+Storage]{\includegraphics[width=0.45\columnwidth, trim={375 85 380 65}, clip]{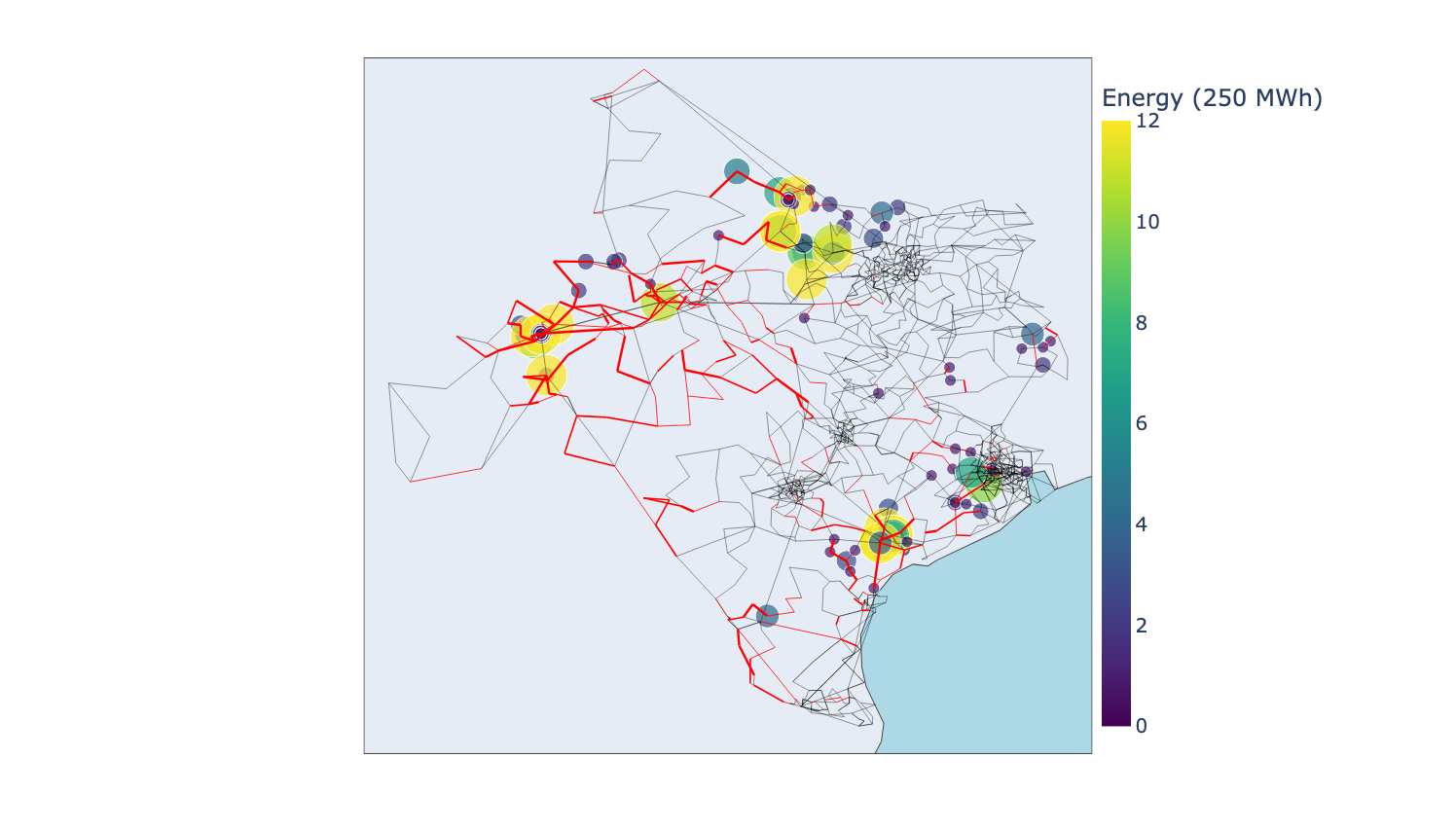}}
\hspace{0.1cm}
\subfloat[2045 TEP+Storage]{\includegraphics[width=0.593\columnwidth, trim={375 85 140 65}, clip]{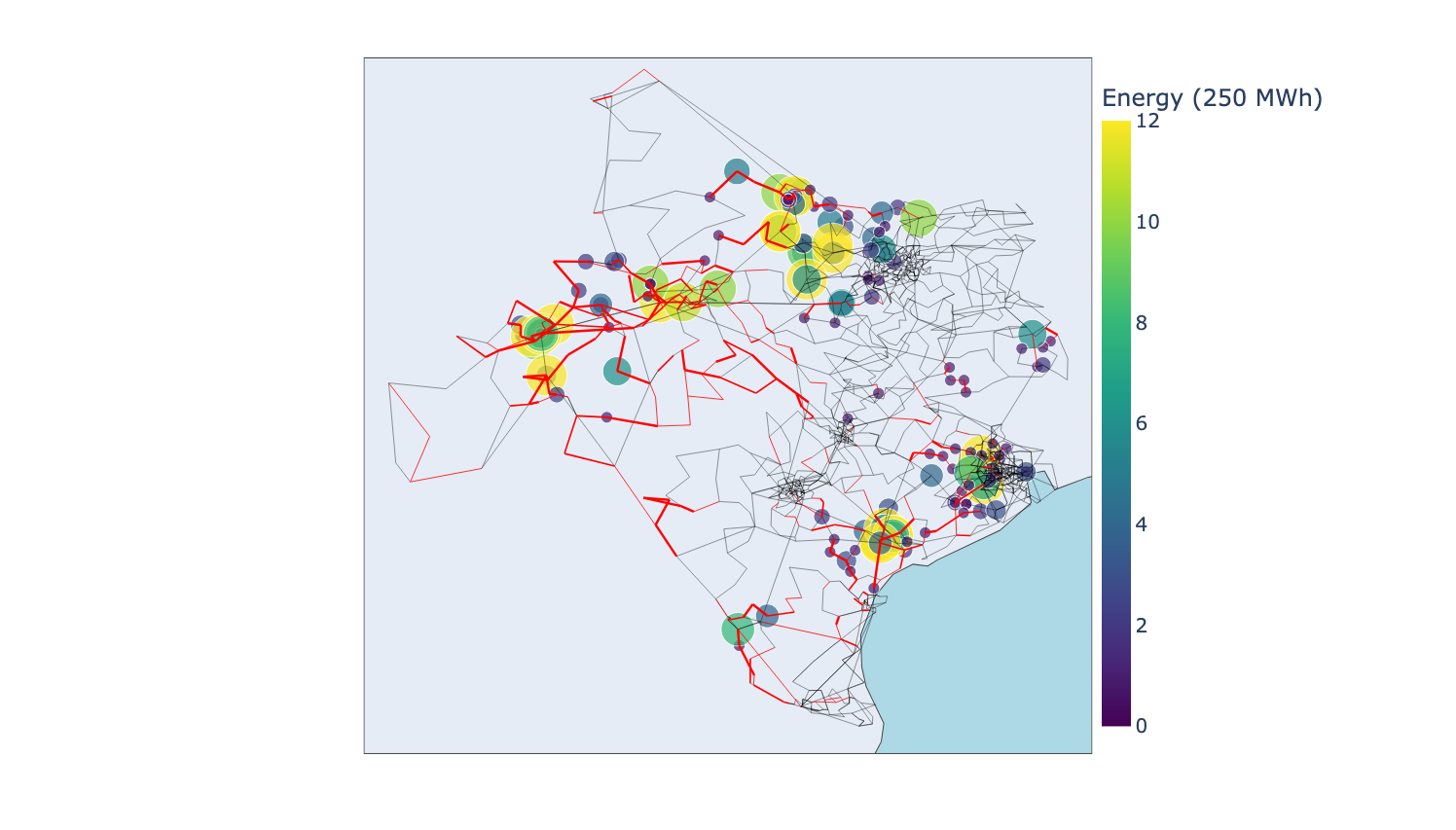}}
\caption{Upgrades from 2030 to 2045 in 5-year intervals. \textcolor{red}{\textbf{Red}} lines denote capacity upgrade locations, with their thickness corresponding to one of three possible upgrade levels. Circles indicate storage locations, with their size and color representing the number of installed units (250MWh each).}
\label{fig:invs}
\end{figure*}

\begin{figure}[!t]
\centering
\includegraphics[width=0.75\columnwidth, trim={0 7 0 7}, clip]{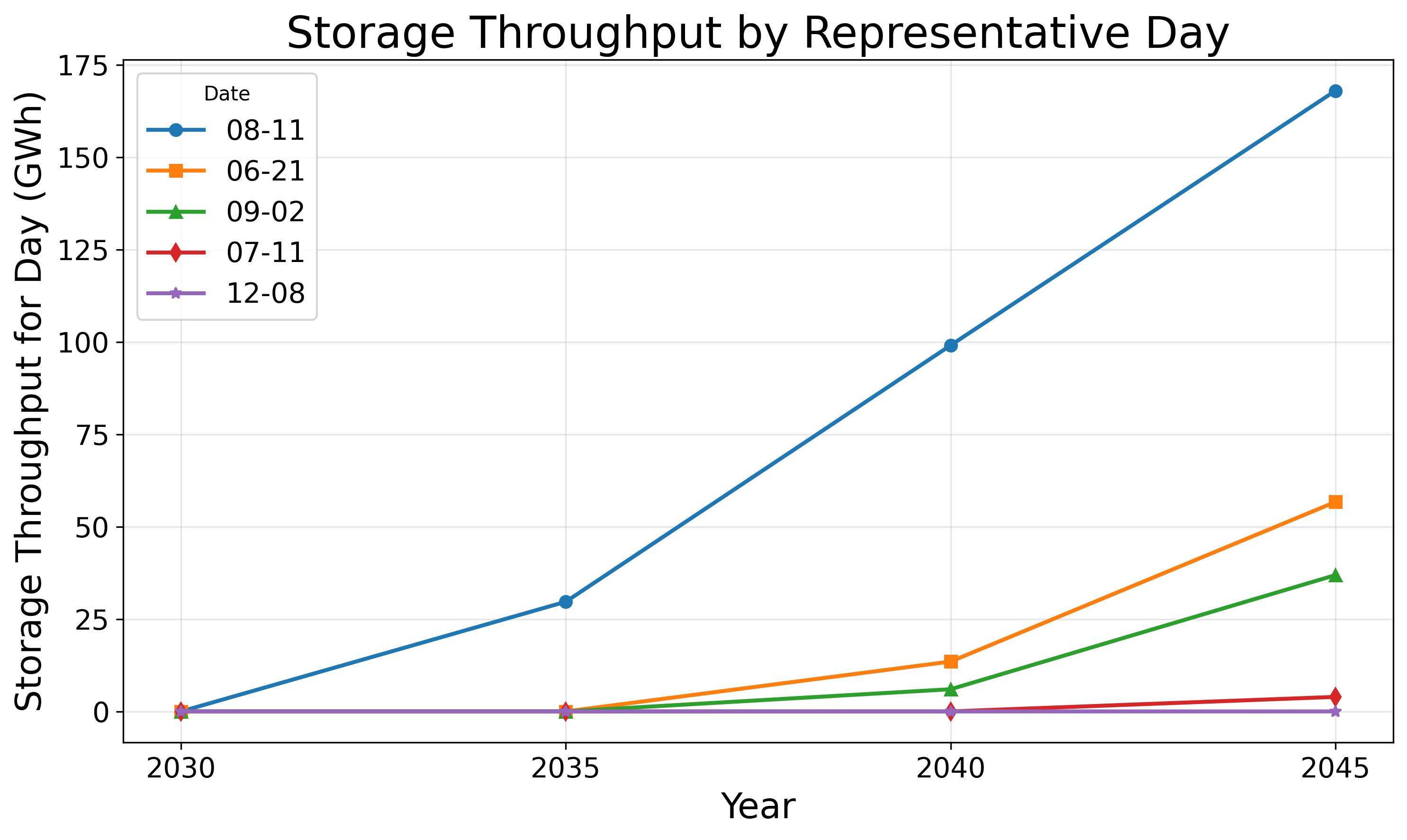}
\caption{Storage energy throughput (GWh) for representative days with nonzero utilization across the 2030–2045 planning horizon. Values reflect the representative day under the investment plan in effect for that year.}
\label{fig:storage_utilization}
\end{figure}

\subsection{Investment Results}
\label{sec:results:investments}
The TR-Benders produces a feasible, near-optimal year-wide plan for 2030–2045 with zero load shed across all representative days. Figure \ref{fig:invs} maps the geographic distribution of investments. Transmission upgrades concentrate in West Texas, strengthening transfer from wind/solar regions, while storage appears beginning in 2035 and grows considerably with increasing load. Figure \ref{fig:storage_utilization} reports storage energy throughputs for each planning year. Only a small fraction of days require storage operation to eliminate load shedding; 08-11 exhibits the largest throughput and aligns with its computational difficulty (Fig. \ref{fig:two_stage_boxplot}). Consequently, multiyear investment choices are largely driven by this day.

Table \ref{tab:investments_opt_gap} summarizes the selected upgrades, storage siting/sizing, and costs; its rightmost column displays the final optimality gaps referenced earlier. The WS results are included to illustrate the quality of the WS procedure prior to applying the TR-Benders algorithm. Additionally, single-day WS results are reported for the most challenging day (08-11) to provide a reference point for the TR-Benders solution and to support the observation above that this day largely drives the resulting investment decisions.

The WS intentionally over-invests. After the TR-Benders, the near-optimal plan reduces $CapEx$ relative to the WS, while incurring only negligible increases in $OpEx$, yielding a lower total cost. The 08-11 single-day results show comparatively higher $OpEx$ than both the WS and TR-Benders plans because they optimize solely for the highest-demand day rather than the equally weighted set of representative days. Because storage dominates investment costs, most of the residual gap is driven by reconciling storage capacity rather than line upgrades. The storage capacity of TR-Benders approaches--but remains above--the requirement implied by the hardest day, consistent with Theorem \ref{thm:first_stage_bound} ensuring zero load-shed feasibility. By 2045, the optimized plan yields 167.75 GWh of storage distributed across 167 nodes, which is equivalent to $32\%$ of the scaled peak renewable capacity. 
\section{Conclusion}
\label{sec:conclusion}
This paper presents a PTDF-based TEP+Storage for joint transmission and storage expansion and applies it to a 2,000-bus synthetic Texas system. To ensure computational feasibility, it introduces a trust-region, multicut Benders scheme warm-started from per-representative-day optima. The method attains final optimality gaps $\leq 2\%$ and remains tractable for a large distributed storage fleet deployed at $\geq 167$ nodes ($\approx 32\%$ of peak renewable capacity). These results demonstrate the proposed PTDF-based methodology effectively handles high spatial resolution and large distributed storage deployments in joint TEP+Storage planning.

Table~\ref{tab:benders_time_iter} shows that computational challenges persist, particularly under high storage penetration and network congestion, and will intensify as temporal, operational, and spatial fidelity increase. Priorities include accelerating SP solves--e.g., pre-screening and dynamic pruning of thermal/flow constraints--and exploring unified Benders variants that combine feasibility and optimality cuts. Another direction is to integrate generation expansion, which may relieve congestion and further improve lazy-constraint performance. Finally, certifying cross-period optimality, testing multistage formulations, and quantifying the value of deferred investments are natural extensions.
\bibliographystyle{IEEEtran}
\bibliography{references.bib}

\end{document}